\documentclass[runningheads]{llncs}
\usepackage[T1]{fontenc}
\usepackage{url}
\usepackage{amsmath}
\usepackage{amssymb}
\usepackage{color}
\usepackage{xspace}
\usepackage{algpseudocode}
\usepackage{algorithm}
\usepackage{graphicx}

\begin{document}
\title{Exact Distributed Sampling}
%
%
\author{Sriram V. Pemmaraju \and
Joshua Z. Sobel}
\authorrunning{S. Pemmaraju and J. Sobel}
%
\institute{Department of Computer Science, University of Iowa, USA
\email{\{sriram-pemmaraju,joshua-sobel\}@uiowa.edu}}
\maketitle              
\newcommand{\josh}[1]{\textcolor{red}{#1}}
\newcommand{\sriram}[1]{\textcolor{blue}{#1}}
\newcommand{\cons}{\mathcal{C}}
\newcommand{\conss}{\mathcal{S}}
\newcommand{\local}{\textsc{Local}\xspace}
\newcommand{\congest}{\textsc{Congest}\xspace}
\newcommand{\congestedclique}{\textsc{CongestedClique}\xspace}

\begin{abstract}
Fast distributed algorithms that output a feasible solution for constraint satisfaction problems, such as maximal independent sets, have been heavily studied. There has been much less research on distributed \textit{sampling} problems, where one wants to sample from a distribution over all feasible solutions (e.g., uniformly sampling a feasible solution).  
Recent work (Feng, Sun, Yin PODC 2017; Fischer and Ghaffari DISC 2018; Feng, Hayes, and Yin arXiv 2018) has shown that for some constraint satisfaction problems there are distributed Markov chains that mix in $O(\log n)$ rounds in the classical \local model of distributed computation.
However, these methods return samples from a distribution close to the desired distribution, but with some small amount of error.  In this paper, we focus on the problem of \textit{exact} distributed sampling. Our main contribution is to show that these distributed Markov chains in tandem with techniques from the sequential setting, namely \textit{coupling from the past} and \textit{bounding chains}, can be used to
design $O(\log n)$-round \local model exact sampling algorithms for a class of weighted local constraint satisfaction problems. 
This general result leads to $O(\log n)$-round exact sampling algorithms that use small messages (i.e., run in the \congest model) and polynomial-time
local computation for some important special cases, such as sampling weighted independent sets (aka the \textit{hardcore} model) and weighted dominating sets.

\keywords{Distributed Sampling, Bounding Chains, Perfect Sampling, Coupling from the Past}
\end{abstract}

\section{Introduction}
There is a vast body of literature on the distributed complexity of solving local constraint satisfaction problems (CSPs) on graphs \cite{AwerbuchLubyGoldbergPlotkinFOCS1989,LinialSICOMP1992,NaorStockmeyerSTOC1993,KuhnMoscibrodaWattenhoferPODC2004,KuhnMoscibrodaWattenhoferJACM2016,BEPSJACM2016,ChangKopelowitzPettieSICOM2019}. 
Here ``local'' refers to the fact that the constraints span vertices with a constant 
diameter in the underlying graph.
These local CSPs include classic ``symmetry breaking'' problems such as maximal independent sets, $(\Delta+1)$-colorings, and maximal matchings \cite{LubySTOC1985,AlonBabaiItaiJAlg1986,GhaffariSODA2016,BEPSJACM2016}. A distributed algorithm solving one of these local CSPs is required to construct some feasible solution of the local CSP.
In contrast, this paper focuses on the problem of \textit{sampling} a feasible solution of a local CSP.  
In the sampling problem the algorithm is required to output a solution sampled from the set of all feasible solutions of the CSP according to some desired probability distribution.
Clearly, the sampling problem is at least as hard as the construction problem because solving the sampling problem requires the construction of a feasible solution.
More precisely, we are interested in sampling solutions of local \textit{weighted} CSPs. Here ``weighted'' refers to an assignment of a weight to each feasible solution of the CSP with the stipulation that the solutions be sampled with probabilities proportional to the weights. When the weights are identical,
the sampling distribution is uniform.

Traditionally, in the sequential setting, sampling from a weighted local CSP on a graph involved running an ergodic Markov chain with a stationary distribution matching the desired distribution.  After running for a long enough time, the distribution of the current state of the chain becomes arbitrarily close to (within any $\epsilon$ in terms of \textit{total variation distance}) the desired distribution.  For a given $\epsilon$, the time required for this is known as the \textit{mixing time}.  After the mixing time is reached the current state of the chain is returned.
The two most simple examples of Markov chains that are used to sample from weighted local CSPs are the \textit{Metropolis-Hastings algorithm} and the \textit{Glauber dynamics} \cite{textbook}.  For some weighted local CSPs (e.g., proper colorings) with certain parameters, these chains have mixing times of $O(n\log \frac{n}{\epsilon})$, where $n$ is the number of nodes in the graph.

A \textit{distributed} Markov Chain for sampling from a weighted local CSP, the \textit{Local Metropolis} chain, was introduced by Feng, Sun, and Yin \cite{LocalMetropolis}.  This chain allows every vertex to simultaneously propose a new label, rather than a single selected vertex.  This chain is easily implemented in the \local model of distributed computing, taking a constant number of rounds for each step of the chain.  For certain weighted local CSPs, the chain can also be implemented in the \congest model, taking a constant number of rounds for each step of the chain.  Both Fischer and Ghaffari \cite{FischerGhaffariColoring} and Feng, Hayes, and Yin \cite{arxivImprovedColoring} showed that this chain could be improved, at least in the case of colorings\footnote{Fischer and Ghffari \cite{FischerGhaffariColoring} claim that their approach, where not every vertex is marked, has an $O(\log \frac{n}{\epsilon})$-round mixing time for a more general class of weighted CSPs than colorings, though the proof of this does not appear in the paper.}, by only updating a small fraction of marked nodes in each step instead of attempting to update every node at every step. 
For colorings, the Local Metropolis chain \cite{LocalMetropolis} has a mixing time of $O(\log \frac{n}{\epsilon})$ when the palette has at least $\alpha \Delta$ colors for any constant $\alpha > 2 + \sqrt{2}$. The improvement in \cite{FischerGhaffariColoring,arxivImprovedColoring} only requires $\alpha>2$, while achieving the same mixing time. Here, $n$ is the number of vertices and $\Delta$ is the maximum degree of the graph.
The point to note about these algorithms is that they return a state drawn from a distribution that approximates the desired distribution within a total variation distance of $\epsilon$. Furthermore, the bound on the mixing time grows as $\epsilon$ becomes smaller. The current paper focuses on exact distributed sampling, i.e., the setting where $\epsilon = 0$.

In the sequential setting, one elegant and well known method for sampling exactly from the stationary distribution of a Markov chain is \textit{coupling from the past (CFTP)} \cite{CFTP}, introduced by Propp and Wilson. In its original form, CFTP took exponential time in general, so its use was limited to Markov chains that had state spaces with special properties (e.g., the \textit{monotonicity} property).  Subsequently, Nelander and Häggström \cite{BoundingChain1} and Huber \cite{BoundingChain2} showed that in cases where the original CFTP algorithm may not be tractable, augmenting the algorithm with \textit{bounding chains} may still allow the algorithm to be used.
The main contribution of this paper is showing that it is possible to use CFTP and the bounding chain technique in the \local model.  We use CFTP and the bounding chain technique in conjunction with the Markov chains described in the previous paragraph, to sample \textit{exactly} from weighted local CSPs.  Our algorithms are fast and in some cases run in the \congest model. 
Our results are described in more detail in the next section.

\paragraph{Comparable Results.} While the papers \cite{LocalMetropolis,FischerGhaffariColoring,arxivImprovedColoring} focus on approximate distributed sampling, there are two recent papers on exact sampling of certain weighted local CSPs.  
Feng and Yin \cite{DistributedJVVPaper} use a seminal result of Jerrum, Valiant, and Vazirani \cite{JerrumValiantVaziraniTCS1986}, showing that for certain problems exact sampling reduces to approximate counting by using a rejection sampling procedure. They present a distributed implementation of that approach.  
Guo, Jerrum, and Liu \cite{LLLPaper} present a sampling version of the Lov\'{a}sz Local Lemma that can be applied in the distributed setting to exactly sample from some weighted local CSPs.

Our approach using CFTP and bounding chains differs significantly from both of these techniques.  All three techniques also differ in terms of the weighted local CSPs they are able to sample from.  However, our techniques lead to results for the hardcore model that improve upon the results of \cite{DistributedJVVPaper,LLLPaper}. A precise comparison appears further below.

\subsection{Main Results}
The main results of our paper can be summarized as follows. 
\begin{enumerate}
    \item We present a distributed Markov chain for sampling from weighted local CSPs
    and prove that it has the correct stationary distribution.  This chain is based on the \textit{Local Metropolis} chain from \cite{LocalMetropolis} and the coloring chains from \cite{FischerGhaffariColoring,arxivImprovedColoring}. We believe that this chain may be the generalization briefly alluded to by Fischer and Ghaffari in \cite{FischerGhaffariColoring}.
    \item We apply the CFTP with bounding chains approach to the above-mentioned distributed Markov chain.  
    We then present a condition that guarantees termination of this algorithm in $O(\log n)$ rounds with high probability in the \local model.
    Thus, under a fairly general condition, we obtain an $O(\log n)$-round algorithm in the \local model for exact sampling from weighted local CSPs. 
    This result is an improvement by a factor of $n$ over the sequential setting running time of $O(n\log n)$, for a slightly different condition, given by \cite{BoundingChain1}. 
    \item   
    We finally show that the general algorithm described above leads to $O(\log n)$-round, small-message  (i.e., \congest model), sampling algorithms for the weighted independent sets (aka hardcore model) problem and the weighted dominating sets problem.  For the hardcore model, we are able to sample within a constant multiple of the hardness threshold for this problem (see the end of this section). Our algorithms do not abuse the power of the \congest model and ensure that every local computation runs in polynomial time. 
    
    The hardcore model is governed by a parameter $\lambda > 0$, called the \textit{fugacity} of the model. Each independent set of size $x$ is assigned a weight of $\lambda^x$; therefore, the desired probability distribution assigns the same probability to all independent sets of the same size. Furthermore, when $\lambda < 1$, small independent sets are more likely than large independent sets.
    Our algorithm\footnote{Results on the hardcore often assume that  algorithms run on graphs of constant degree.  In this sense, we require $\lambda<\frac{1}{\Delta}$.  Furthermore, on graphs of bounded degree, our results can likely be slightly improved by using the LubyGlauber chain from \cite{LocalMetropolis}.} for the hardcore model requires $\lambda\leq\frac{\alpha}{\Delta}$, for any constant $\alpha<1$. This almost matches the condition $\lambda\leq\frac{\alpha}{\Delta-1}$ that \cite{BoundingChain1} gives for $O(n\log n)$ time in the sequential setting.  We also derive a similar condition for weighted dominating sets.
    
    Feng and Yin \cite{DistributedJVVPaper} also present results for exactly sampling from the hardcore model. Their algorithm is much slower than ours, taking $O(\log^3 n)$ rounds, and it also uses large messages and exponential-time local computations. However, their result holds for a wider range of the fugacity parameter, specifically when $\lambda < \frac{(\Delta-1)^{\Delta-1}}{(\Delta-2)^\Delta}$. Note that from \cite{LocalMetropolis}, sampling from the hardcore model is hard in the \local model for $\lambda > \frac{(\Delta-1)^{\Delta-1}}{(\Delta-2)^\Delta}$.  In the sequential setting, the same threshold is a barrier between polynomial and non-polynomial sampling, unless RP=NP \cite{SequentialLowerBound1,SequentialLowerBound2,SequentialLowerBound3,SequentialLowerBound4}, given the connection between approximate counting and sampling \cite{JerrumValiantVaziraniTCS1986,countingsampling1}.  Like us, Guo, Jerrum, and Liu \cite{LLLPaper} provide an $O(\log n)$-round w.h.p. \congest algorithm for exact sampling from the hardcore model using polynomial-time local computation; however, they require a smaller range of $\lambda$ than us, specifically $\lambda \leq \frac{1}{2\sqrt{e}\Delta-1}$.  
    
    Our algorithm improves over both \cite{DistributedJVVPaper,LLLPaper}, as every vertex always outputs its label in the exact sample.  The algorithm from \cite{DistributedJVVPaper} succeeds with high probability and returns an exact sample conditioned on success; however, it may fail and failures cannot be detected by every vertex locally.  On the other hand, the algorithm from \cite{LLLPaper} always succeeds in a random amount of time like ours; however, a vertex cannot locally determine when its portion of the output is finalized.

\end{enumerate}

\section{Technical Preliminaries}
\subsection{Sampling Weighted Local CSPs}
A weighted CSP on a graph $G = (V,E)$ consists of a set $L$ of vertex labels and a collection $\conss \subseteq 2^V$ of \textit{constraint sets}, in addition to a \textit{constraint} $C_R$ for each constraint set $R$ .  
A labeling assigns an element in $L$ to each vertex in $V$; thus $L^V$ is the set of all labelings of $G$. 
For a labeling $\ell \in L^V$, $v \in V$, and $R \subseteq V$, we use $\ell(v)$ to denote the label that $\ell$ assigns to vertex $v$ and $\ell\restriction_R$ to denote the restriction of $\ell$ to $R$.
Each constraint $C_R$ maps the set of all restricted labelings $\ell\restriction_{R}$ to the non-negative reals.
The \textit{weight} of labeling $\ell$ is  
$\prod_{R\in \conss} C_R(\ell\restriction_R)$.
We call a labeling \textit{valid} if it has weight greater than zero.  
A \textit{CSP} is typically defined as the problem of finding an arbitrary valid labeling. The \textit{weighted CSP} is to sample valid labelings, where the probability of choosing a labeling is proportional to its weight.
This probability distribution of labelings will be referred to as $\pi$.

For notational convenience, we will assume that $\conss$ does not contain any singleton sets. Instead, we will assume that for each vertex $v$ there is a separate unary constraint $b_v: L\to \mathbb{R}^{\geq 0}$, that maps a label assigned to $v$ to a non-negative real number.  Note that this is without loss of generality because for any $v \in V$, we can set $b_v(x) = 1$ for all $x \in L$. With this additional notation, the weight of a labeling $\ell$ can be written as
\begin{equation}
\label{eqn:weight}
\prod_v b_v(\ell(v)) \cdot \prod_{R\in \conss} C_R(\ell\restriction_R).
\end{equation}
When there is a constant $k$ such that every constraint set has a diameter in the graph $G$ bounded by  $k$, then the weighted CSP is called a \textit{weighted local CSP}. Note that the diameter here refers to the distance in $G$, not the distance in the subgraph of $G$ induced by $R$.

We consider three examples of weighted local CSPs in this paper:  weighted independent sets (the hardcore model), weighted dominating sets, and (briefly) the Ising model.
\begin{itemize}
    \item \textit{Weighted independent sets} are a weighted local CSP taking a parameter $\lambda>0$ commonly called the \textit{fugacity}. Here, the set of labels $L = \{0, 1\}$, the vertices labeled $1$ form the independent set. Each unary constraint $b_v$ is the function $b_v(0) = 1, b_v(1)=\lambda$. The collection $\conss$ of constraint sets is $E$, the set of all edges in the graph. For every edge $e = \{u,v\}$, the constraint $C_e(\ell(u), \ell(v)) = 0$ if $\ell(u) = \ell(v) = 1$; otherwise, $C_e(\ell(u), \ell(v)) = 1$.
    $C_e$ is simply ensuring that the valid labelings are independent sets of $G$. Here an independent set of size $x$ has weight $\lambda^x$. Thus all independent sets of the same size have uniform probability; however, for $\lambda<1$ small independent sets have a higher probability than large independent sets.
    \item \textit{Weighted dominating sets} are a weighted local CSP very similar to weighted independent sets. This CSP also takes a parameter $\lambda>0$ and has the same set of labels and unary constraints. The collection $\conss$ of constraint sets consists of inclusive neighborhoods $N_v = Nbr^+(v)$ for each vertex $v$. The constraint $C_{n_v}$ maps to $1$ if at least one of the vertices in the inclusive neighborhood has the label $1$ and $0$ otherwise. 
    \item The \textit{Ising model}, in a simple form as given by \cite{LocalMetropolis}, is another weighted local CSP.  Here there are two possible labels, $\{-1,1\}$, for each vertex, and $b_v(-1)=b_v(1)=1$.  $\beta>0$ is provided with the model.  Each pair of adjacent vertices has a constraint that maps to $\beta$ if the vertices have the same label and $1$ otherwise.  The weight given to a labeling is $\beta^a$, where $a$ is the number of edges that have both of their vertices assigned the same label.
\end{itemize}


\subsection{Related Work on lower bounds}
The discussion thus far has been on upper bounds. However, there are interesting lower bounds for sampling from the hardcore model in the \local model \cite{LocalMetropolis,LLLPaper}. 
Feng, Sun, and Yin show an $\Omega(Diam)$ round lower bound on $n$-vertex graphs with diameter $Diam = \Omega(n^{1/11})$, when $\lambda > \frac{{(\Delta-1)^{\Delta-1}}}{(\Delta-2)^\Delta}$.  Guo, Jerrum, and Liu \cite{LLLPaper} show a more general $\Omega(\log n)$-round lower bound.  Similar bounds exist or can be derived for other weighted local CSPs, see the two cited papers.

\section{Distributed Markov Chain}
\label{section:metropolis}
In this section, we present a distributed Markov chain for sampling from weighted local CSPs (see Algorithm \ref{markov-alg}). 
In subsequent sections we show that it is possible to use this chain to exactly sample solutions of weighted local CSPs efficiently in the \local model (and in the \congest model in some cases).
Our Markov chain is a simple modification of the \textsc{LocalMetropolis} Markov chain, given in \cite{LocalMetropolis}.  
The \textsc{LocalMetropolis} chain contains a \textit{propose} step in which each vertex $v$ independently proposes a label $\sigma_v$ in $L$ with probability proportional to 
$b_v(\sigma_v)$. This is followed by a probabilistic \textit{local filter} that is applied to each constraint set. We describe this in more detail in the next paragraph.
If all the constraint sets containing a vertex $v$ pass the local filter, then $v$ adopts $\sigma_v$; otherwise it retains its old label.
We modify this Markov chain by first marking each vertex independently with a fixed probability $p$ and then allowing only the marked vertices to be active in each step.
This idea -- of randomly sampling vertices which will be active -- is a standard idea in randomized distributed algorithms (for example, Luby's algorithm \cite{AlonBabaiItaiJAlg1986,LubySTOC1985}), but more to the point it was used in \cite{FischerGhaffariColoring,arxivImprovedColoring} to speed up their distributed Markov chain for coloring.  In particular, \cite{FischerGhaffariColoring,arxivImprovedColoring} both present the chain resulting from augmenting the \textsc{LocalMetropolis} chain for colorings to have a set of marked vertices.
Furthermore, we infer that the Markov chain we present is the generalization mentioned by \cite{FischerGhaffariColoring}.

We now describe the local filtering step.
During each step of the chain, if a vertex $v$ is marked then it has a current label $X_v$ and a proposed label $\sigma_v$; otherwise, it only has a current label $X_v$.  For each constraint set $R$, we now consider a collection $\mathcal{L}(R)$ of labelings of $R$.  In particular, we call $ (\ell(v_1), \ell(v_2),\ldots,\ell(v_{|R|}))$ a \textit{potential} labeling of $R$, if each label $\ell(v_i)$ is chosen from either $X_{v_i}$ or $\sigma_{v_i}$, and as long as at least one of the $|R|$ choices was made from $\sigma$.  We now let $\mathcal{L}(R)$ be the collection of all potential labelings, with the note that it can contain the same element multiple times if there are multiple valid ways of choosing it.  To be technically correct, each potential labeling should be represented as a binary vector, however we bend notation and treat a potential labeling as a labeling of the constraint set.  Each constraint set $R$ passes the local filter with probability 
\begin{equation}
\label{eqn:localfilter}
\prod_{\ell\in \mathcal{L}(R)}\frac{C_R(\ell)}{C^*_R}.
\end{equation}  
Here, $C^*_R$ refers to the maximum value that the constraint can take over its entire domain $L^R$.  $C^*_R$ can be assumed to be nonzero; otherwise, every labeling would be invalid.

\begin{figure}
    \centering
    \includegraphics[width=\textwidth]{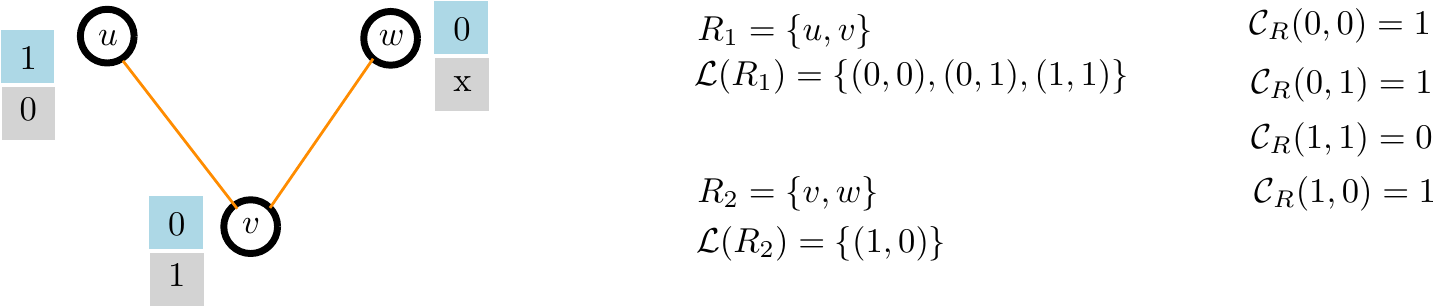}
    \caption{Example illustrating the computation of local filter probabilities for the weighted independent set problem.}
    \label{fig:localFilterExample}
\end{figure}

Figure \ref{fig:localFilterExample} shows an example of how to apply the local filter to the weighted independent set problem.
The figure shows a subgraph with 3 vertices and alongside each vertex we show its current label (top) and proposed label (bottom). Vertex $w$ has ``x'' as its proposed label because it is not marked active in this iteration.
Vertex $v$ participates in two constraint sets $R_1$ and $R_2$ corresponding to the two edges incident on it. The potential labelings $\mathcal{L}(R_1)$ and $\mathcal{L}(R_2)$ of each constraint set are shown. Each tuple in $\mathcal{L}(R_1)$ (respectively, $\mathcal{L}(R_2)$) shows $u$'s label (respectively, $v$'s label) followed by $v's$ label (respectively, $w$'s label).
Note that tuple $(1, 0)$ does not appear in $\mathcal{L}(R_1)$ because that would correspond to to both labels being chosen from the current labels.
Finally, the values of the constraint function are shown on the right. In this case we use the same constraint function for every constraint set. Note that the expression in (\ref{eqn:localfilter}) evaluates to 0 for $R_1$ indicating that $R_1$ does not pass the local filter. For $R_2$ the local filter evaluates to 1, indicating that $R_2$ does pass the local filter. Since not all constraint sets of $v$ pass the local filter, $v$ does not update its label.

\begin{algorithm}
\caption{}
\label{markov-alg}
\begin{algorithmic}
\Require each vertex $v$ initially has label $X_v$
\State each $v\in V$ is marked active with probability $p$
\For{each active vertex $v\in V$} 
    \Comment{Propose step}
    \State $v$ chooses $\sigma_v\in L$ with probability proportional to $b_v(\sigma_v)$
\EndFor
\For{each constraint set $R \in \conss$}
    \Comment{Local Filter step}
    \State $R$ passes the local filter with probability $\prod_{\ell\in \mathcal{L}(R)}\frac{C_R(\ell \restriction R)}{C^*_R}$
\EndFor
\For{each active vertex $v\in V$} 
    \Comment{Finalize labels}
    \If{all constraint sets in $\conss$ containing $v$ pass their checks}
        \State $X_v=\sigma_v$
    \EndIf
\EndFor
\end{algorithmic}
\end{algorithm}

We next prove that the Markov chain given by Algorithm \ref{markov-alg} is ergodic (aperiodic and irreducible) with $\pi$ as its stationary distribution under a mild condition; for any two valid states $X$ and $Y$, we have a sequence of valid states $X=Z_1,...,Z_n=Y$, where adjacent states differ at only a single vertex.  Recall that $\pi$ is the distribution over labelings of $G$ in which each labeling $\ell$ has probability proportional to its weight, expression (\ref{eqn:weight}).  This proof is based on and an extension of the proof from \cite{LocalMetropolis}.
The chain is clearly aperiodic since every state can transition to itself.  Furthermore, the condition mentioned above ensures irreducibility.  It remains to show that $\pi$ is the stationary distribution.  A standard way of showing that a distribution $\pi$ is stationary for a Markov chain with transition matrix $M$ is to prove the \textit{detailed balance equations},

$$\pi(X) \cdot M[X, Y]= \pi(Y) \cdot M[Y, X]$$ for all states of the chain $X, Y$.
As noted in \cite{LocalMetropolis}, a slightly stronger condition is needed for the chain to have the correct limit distribution if it starts in a state with weight $0$.  This point will not concern us, as we will view the chain as only being defined over valid states.  

\begin{theorem}
\label{detailed-balance}
Let $M$ be the transition matrix of the Markov chain defined by Algorithm \ref{markov-alg}.
For all states $X, Y$ we have $\pi(X) \cdot M[X, Y] = \pi(Y) \cdot M[Y, X]$.
\end{theorem}
\begin{proof}
Let $X \neq Y$ be two valid labelings, thus $X$ and $Y$ are also states in the Markov chain.
Let $X_v$ be the current label of vertex $v$ and $\sigma_v$ be the proposed label if the vertex is marked, otherwise let $\sigma_v=\sqcup$.  We also define the binary vector $\mathbf{I}$, where for each constraint set $T \in \mathcal{S}$, $\mathbf{I}_T = 1$ if the constraint set $T$ passes its check and $\mathbf{I}_T = 0$ otherwise.  If a vertex $v$ is contained in any constraint set $T$ such that $\mathbf{I}_T = 0$, we call $v$ \textit{restricted}.  Note that $(\sigma, \mathbf{I})$ defines a function from one state to another.  

Given a pair of states, $X$ and $Y$, there could be many tuples $(\sigma, \mathbf{I})$ that map $X$ to $Y$ and similarly many tuples $(\sigma', \mathbf{I}')$ that map $Y$ to $X$.
We now show a bijection between tuples mapping $X$ to $Y$ and tuples mapping $Y$ to $X$.  Let $(\sigma, \mathbf{I})$ be a tuple mapping $X$ to $Y$.  
From $(\sigma, \mathbf{I})$, we construct $(\sigma', \mathbf{I}')$, a function mapping $Y$ to $X$ as follows.
For every vertex where $Y_v\neq X_v$, let $\sigma'_v = X_v$.  For all other vertices, let $\sigma'_v = \sigma_v$.  Finally, let $\mathbf{I}=\mathbf{I}'$.
To see that $(\sigma', \mathbf{I}')$ maps $Y$ to $X$, first note that for any vertex $v$ where $X_v\neq Y_v$, we must have $\sigma_v = Y_v$ and $v$ must be unrestricted in $\mathbf{I}$. This means that $v$ is unrestricted in $\mathbf{I}'$ and since $\sigma'_v = X_v$, we see that $(\sigma', \mathbf{I}')$ maps $Y_v$ to $X_v$.
For a vertex $v$ where $X_v=Y_v$, either (a) $\sigma_v = Y_v = X_v$, (b) $\sigma_v = \sqcup$, or (c) $v$ is restricted in $\mathbf{I}$.
In case (a), $\sigma'_v = Y_v = X_v$, in case (b) $\sigma'_v = \sqcup$, and in case (c) $v$ is restricted in $\mathbf{I}'$. In all three cases, $(\sigma', \mathbf{I}')$ maps $Y_v$ to $X_v$.
It can be checked that this construction gives a bijection $(\sigma,\mathcal{C})\leftrightarrow (\sigma', \mathcal{C}')$ between tuples that map from $X$ to $Y$ and those that map from $Y$ to $X$.  

It is now sufficient to show 
\begin{equation}
\label{eqn:sufficient}
\frac{P(\sigma)P(\mathbf{I} | \sigma, X))}{P(\sigma')P(\mathbf{I}' | \sigma', Y))}=\frac{\pi(Y)}{\pi(X)}.
\end{equation}
This is because
$$\pi(X) \cdot M[X, Y] =  \pi(X) \sum P(\sigma)P(\mathbf{I} | \sigma, X)) = \pi(Y) \sum P(\sigma')P(\mathbf{I}' | \sigma', X)) = \pi(Y) M[Y, X],$$ 
where the middle equality follows from (\ref{eqn:sufficient}).

To show (\ref{eqn:sufficient}), we first observe that
$$\frac{P(\sigma)}{P(\sigma')}=\prod\limits_{v | X_v \neq Y_v}\frac{p \cdot b_v(Y_v)}{p \cdot b_v(X_v)} = \prod\limits_v\frac{b_v(Y_v)}{b_v(X_v)}.$$
We now consider $\frac{P(\mathbf{I} | \sigma, X))}{P(\mathbf{I}' | \sigma', Y))}$.  Since each constraint is passed or failed independently,
 $$\frac{P(\mathbf{I} | \sigma, X))}{P(\mathbf{I}' | \sigma', Y))}=\prod\limits_{T\in \conss}\frac{P(\mathbf{I}_T | \sigma, X))}{P(\mathbf{I}_T' | \sigma', Y))}.$$

There are now two cases to consider.
\begin{description}
\item[Case $\mathbf{I}_T = 0$.] In this case, every vertex in the constraint set $T$ is restricted.  For each of these vertices, we must have $X_v=Y_v$ and also $\sigma_v = \sigma'_v$.  
This means that the set of potential labelings $\mathcal{L}(T)$ used in the local filter probability (\ref{eqn:localfilter}) are identical for the chain in state $X$ and the chain in state $Y$.
Thus
$P(\mathbf{I}_T = 0 | \sigma, X) = P(\mathbf{I}_T' = 0 | \sigma', Y)$. 
Therefore, we can can rewrite the ratio $P(\mathbf{I}_T = 0 | \sigma, X)/P(\mathbf{I}_T' = 0 | \sigma', Y)$ as
$$\frac{P(\mathbf{I}_T = 0 | \sigma, X))}{P(\mathbf{I}_T' = 0 | \sigma', Y))} = 1 = \frac{C_T(Y_{v_1}, \ldots, Y_{v_{|T|}})}{C_T(X_{v_1}, \ldots, X_{v_{|T|}})}.$$

\item[Case $\mathbf{I}_T = 1$.]
We establish a mapping between potential labels in set $\mathcal{L}(T)$ for state $X$ and 
potential labels in set $\mathcal{L}(T)$ for state $Y$ that is almost a bijection.
Let $\mathcal{X} = (\mathcal{X}_{v_1}, \mathcal{X}_{v_2}, \ldots,\mathcal{X}_{v_{|T|}})$ be a potential labeling for state $X$.
From $\mathcal{X}$, we create a potential labeling
$\mathcal{Y} = (\mathcal{Y}_{v_1}, \mathcal{Y}_{v_2} \ldots,\mathcal{Y}_{v_{|T|}})$
for state $Y$, as follows.
\begin{itemize}
\item If $\mathcal{X}_{v_i}$ was chosen from $X_{v_i}$ and $X_{v_i} = Y_{v_i}$ we can let $\mathcal{Y}_{v_i}$ be chosen from $Y_{v_i}$. 
\item If $\mathcal{X}_{v_i}$ was chosen from $X_{v_i}$ and $X_{v_i} \neq Y_{v_i}$ we can let $\mathcal{Y}_{v_i}$ be chosen from $\sigma'_i$.  
Note that in this case, $\sigma'_{v_i} = X_{v_i}$.
\item If $\mathcal{X}_{v_i}$ was chosen from $\sigma_{v_i}$ and $X_{v_i} = Y_{v_i}$ we can let $\mathcal{Y}_{v_i}$ be chosen from $\sigma'_{v_i}$.  
\item If $\mathcal{X}_{v_i}$ was chosen from $\sigma_{v_i}$ and $X_{v_i} \neq Y_{v_i}$ we can let $\mathcal{Y}_{v_i}$ be chosen from $Y_{v_i}$.  
\end{itemize}
Note that in all 4 cases, $\mathcal{X}_{v_i} = \mathcal{Y}_{v_i}$.

Note that $\mathcal{Y}$ is a potential labeling from $\mathcal{L}(T)$ as long as not every choice was from $Y_{v_i}$.  
There is exactly one choice for the labels in
$\mathcal{X} = (\mathcal{X}_{v_1}, \mathcal{X}_{v_2}, \ldots,\mathcal{X}_{v_{|T|}})$
that is mapped to
$\mathcal{Y} = (Y_{v_1}, Y_{v_2}, \ldots, Y_{v_{|T|}})$. In this choice, every $\mathcal{X}_{v_i}$ was chosen from $\sigma_{v_i}$ when $X_{v_i}\neq Y_{v_i}$ and $X_i$ when $X_{v_i}=Y_{v_i}$.  
We denote this label $\mathcal{X}'$ and see that $\mathcal{X}'=(Y_{v_1}, Y_{v_2}, \ldots, Y_{v_{|T|}})$.  Furthermore, since $X\neq Y$ at least one choice for $\mathcal{X}'$ was made from $\sigma$ so $\mathcal{X}'$ is a valid labeling.

Now recall that $P(\mathbf{I}_T = 1|\sigma, X) = \prod_{\ell\in \mathcal{L}(T)}\frac{C_T(\ell \restriction T)}{C^*_T}$ according to (\ref{eqn:localfilter}). Thus, in the ratio $P(\mathbf{I}_T = 1|\sigma, X)/P(\mathbf{I}'_T = 1|\sigma, Y)$, all terms in the numerator cancel out except for $C_T(Y_{v_1}, Y_{v_2}, \ldots, Y_{v_{|T|}})$.
By a symmetric argument, there is a single term $C_T(X_{v_1}, X_{v_2}, \ldots, X_{v_{|T|}})$ left in the denominator.
\end{description}
Altogether, this shows  $$\frac{P(\mathbf{I}_T = 1 | \sigma, X))}{P(\mathbf{I}_T' = 1 | \sigma', Y))} = \frac{C_T(Y_{v_1}, ..., Y_{v_{|T|}})}{C_T(X_{v_1}, ..., X_{v_{|T|}})},$$ completing the proof.
\end{proof}

\section{Distributed Coupling From the Past}

\subsection{Coupling From the Past}
\label{section:CFTP}
\textit{Coupling from the past (CFTP)}, introduced in \cite{CFTP}, is a technique for sampling \textit{exactly} from the stationary distribution of an ergodic Markov chain.  A chapter covering coupling from the past can be found in \cite{textbook}.  Suppose the Markov chain is defined over a set of states $\Omega$ and has a transition matrix $M$.  
Following the notation in \cite{Vigoda}, 
let $f : \Omega \times \{0, 1\}^* \to \Omega$ be a function such that $P(f(X, r) = Y)= M[X, Y]$, when $r$ is a string chosen uniformly at random from $\{0, 1\}^*$. 
The function $f$ is called a \textit{random mapping} representation of the transition matrix $M$ and is known to always exist (see Proposition 1.5 in \cite{textbook}).

The function $f$ allows us to use a common source of randomness and evolve a chain beginning from each state $\sigma \in \Omega$ in a ``coupled'' manner. 
More precisely, for integer $t$, let $r_t \in \{0, 1\}^*$ be chosen uniformly at random.
Let $f_t : \Omega \to \Omega$ be the function $f_t(X) = f(X, r_t)$.  Now define a function $F^{t_2}_{t_1}: \Omega \to \Omega$ for integers $t_2 \geq t_1$ as

$$F^{t_2}_{t_1}(X) = (f_{t_2-1} \circ f_{t_2 - 2} \circ \cdots \circ f_{t_1})(X) = f_{t_2-1}(f_{t_2-2}(\ldots f_{t_1}(X))).$$
The function $F^{t_2}_{t_1}$ defines a coupled evolution of states from time $t_1$ to time $t_2$ with the property that $P(F^{t_2}_{t_1}(X) = Y) = M^{t_2-t_1}[X, Y]$.  For each $\sigma \in \Omega$, $F_0^{t}(\sigma)$ defines the state at time $t$ of a Markov chain beginning at $\sigma$ and with transition matrix $M$. The entire collection of Markov chains that evolve in this manner, one starting in each state of $\Omega$ and using common random strings, is called a \textit{grand coupling} \cite{textbook}.

Now consider the function
$$F^{0}_{-T}(X) = (f_{-1} \circ f_{-2} \circ \cdots \circ f_{-T})(X) = f_{-1}(f_{-2}(\ldots f_{-T}(X))).$$
The insight of the CFTP technique is that if we can show that with probability $1$ there exists a $T$ such that $F^0_{-T}$ is a constant function (i.e., every state in $\Omega$ is mapped to the same state by $F^0_{-T}$), 
then the unique element in this image is drawn exactly from the stationary distribution $\pi$ of the Markov chain. We now provide some intuition for this possibly surprising claim. 
First note that if $F^0_{-T_1}$ maps all elements in $\Omega$ to $\omega\in\Omega$, then for any $T_2 > T_1$, the function $F^0_{-T_2}$ also maps all elements in $\Omega$ to $\omega$.  
Note that this is not true for the symmetric and invalid sampling technique, ``coupling to the future'', where we use $F_0^T$ instead of $F_{-T}^0$.
The critical difference here is that even if $F_0^T$ becomes constant for a sufficiently large $T$, we may not have $F_0^{T+1}=F_0^T$. 
We can therefore intuitively think of CFTP as computing $F^0_{-\infty}$. Since the stationary distribution is the limit distribution of an ergodic Markov chain, the unique element in the image of $F^0_{-\infty}$ is drawn exactly from the stationary distribution.

This insight suggests a natural algorithm for exact sampling. Starting with $T'=1$, compute $F^0_{-T'}(\Omega)$ and check if $|F^0_{-T'}(\Omega)| = 1$. If it is, we output the unique element (state) in $F^0_{-T'}(\Omega)$, otherwise we double the value of $T'$ and repeat.  Propp and Wilson point out that by choosing to double $T'$ at each iteration we only overshoot the smallest value of $T$ where $|F^0_{-T}(\Omega)| = 1$ by a constant multiple \cite{CFTP}.
To avoid biasing the samples it is critical that the same choice of functions $f_{-1}, f_{-2}, \ldots$ is used to compute $F^0_{-T'}$ each time $T'$ is increased and that the process is not stopped even if $T'$ grows large without $F^0_{-T'}$ becoming constant.

Given that $|\Omega|$ can be exponentially large relative to the size of the input, a problem with this algorithm is efficiently checking if $F^0_{-T'}$ is a constant function.  Another issue is that in general the final value of $T'$ needed may be large .  

\subsection{Bounding Chains}
\label{section:boundingChains}

It is easy to check if $F^0_{-T}$ is a constant function in the special case where the grand coupling has a property called \textit{monotonicity} \cite{CFTP}; however, we want to sample from weighted local CSPs where monotonicity may not exist.  
Häggström and Nelander \cite{BoundingChain1} and Huber \cite{BoundingChain2}
describe the \textit{bounding chain} technique for determining when coupling from the past gives a constant function.  The idea, is that for each vertex we compute a superset of the labels it could be assigned by $F_{-T}^0$ for an unknown valid input labeling.  If this label set is a singleton for every vertex, $F_{-T}^0$ is a constant function.  Note that the converse does not necessarily hold.  When the label set of every vertex is a singleton, we call $F^0_{-T}$ singular (acknowledging the slight abuse of notation, since singularity also depends on the method of computing the label sets).  We can implement CFTP by checking whether $F^0_{-T}$ is singular, instead of constant.  Note that a trivial choice for each label set is $L$, however, this will lead to an algorithm that never terminates.

To give a concrete practical example, consider sampling from the hardcore model in the sequential setting using bounding chains \cite{BoundingChain1}.  One step of the Markov chain is defined as follows.  A uniformly random vertex $v$ is selected and a coin is flipped with probability of heads $\frac{\lambda}{1+\lambda}$.  If the coin flip is tails then $v$ is removed from the independent set. If the coin flip is heads and if no neighbors of $v$ are in the current independent set, then $v$ joins the independent set. The random mapping representation of this Markov chain, $f$, can be defined in a natural way with the random string $r$ encoding a vertex and a coin flip, chosen with the correct distribution from all $(\text{vertex}, \text{coin flip})$ pairs. 

We now compute a superset of labels each vertex can be assigned by $F_{-T}^0$.  We use a recursive approach.  In the trivial base case, $T=0$, every vertex receives the label set $\{0,1\}$, indicating that in some states the vertex is outside the independent set and in some states the vertex is inside the independent set.  For $T\geq 1$, we first compute the label set for every vertex for $F_{-T+1}^0$.  If the operation of $f_{-T}$ is the removal of vertex $v$ from the independent set, we can be sure that regardless of the input state, $v$ has the label $0$.  Thus v's set of possible labels is set to $\{0\}$.  On the other hand, when $f_{-T}$ represents $v$ attempting to join the independent set, there are a few cases.  If every neighboring vertex of $v$ has the label set $\{0\}$ from the recursive call, then $v$ joins the independent set regardless of the input state, and is assigned the label set $\{1\}$.  If any neighboring vertex of $v$ has the label set $\{1\}$ from the recursive call, the join will fail regardless of the input state, and $v$ will get the label set $\{0\}$.  In the remaining case, whether the join succeeds may depend on the input state.  Thus we assign $v$ the label set $\{0,1\}$.  Every other vertex simply keeps its set from the recursive call.  It is shown in \cite{BoundingChain1}, that singularity occurs for this process for $T=O(n\log n)$ in expectation, as long as the fugacity $\lambda\leq\frac{\alpha}{\Delta-1}$, for any constant $\alpha<1$.  This also holds with high probability.

\subsection{Distributed Bounding Chains}
This section contains our main result in which we show how to apply CFTP and the bounding chain technique to the distributed Markov chain defined in Algorithm \ref{markov-alg}.
This yields Theorem \ref{runtimetheorem} which shows that if the weighted local CSP satisfies a general condition we can sample exactly in $O(\log n)$ rounds in the \local model.  With additional conditions we get an $O(\log n)$-round \congest algorithm.  Note that on some weighted local CSPs, such as colorings, this algorithm will run forever.  
Furthermore, as stated earlier, terminating long running instances of the algorithm before a sample is returned may bias the results.

We first note that there is a clear choice of a random mapping representation of the chain described by Algorithm \ref{markov-alg}.  In Algorithm \ref{markov-alg}, three random choices are made: (i) each $v \in V$ is marked active with probability $p$, (ii) each active vertex $v$ picks a label $\sigma_v$ with probability proportional to $b_v(\sigma_v)$, and (iii) each constraint set $R$ passes the local filter with a certain probability (defined in Equation (\ref{eqn:localfilter})).  These random choices can be collectively specified by string $r \in \{0, 1\}^*$ chosen uniformly at random.  Since we want to execute this algorithm in the \local model, we note that $r_t$ can be generated in a distributed fashion.  Each vertex $v$ can locally pick a random bit specifying whether it is marked active and choose $\sigma_v$ at random from $L$ proportional to its $b_v$ values. To determine if the constraint set $R$ passes the local filter, we could have one vertex in $R$ (e.g., the vertex with the highest ID in $R$) pick a number uniformly at random from $[0, 1]$.  To avoid worrying about precision, we actually let this vertex generate a binary table (with arbitrary precision) specifying whether the constraint will pass or not given every possible combination of current labels and proposals.  

The algorithm starts with $T' = 1$ and every vertex having label set $\{l\in L: b_v(l)>0\}$. The algorithm then proceeds in \textit{stages} $1, 2, \ldots$.  After each stage, $T'$ is doubled.  At the start of a stage, some vertices have already output a label; we call these vertices \textit{coalesced} and the remaining vertices \textit{uncoalesced}.  Each stage is initiated by uncoalesced vertices. The coalesced vertices are in a ``stand by'' mode for the stage and will only become active if and when prompted by uncoalesced vertices.  At the end of a stage, any uncoalesced vertices that now have singleton label sets output the unique label in their label sets.  The algorithm is complete after every vertex outputs a label.

Each stage consists of two phases, a \textit{preprocessing} phase and a \textit{main} phase. We now describe both of these phases separately.  Recall that $k$ is the largest diameter for any constraint.

\noindent
\textbf{Preprocessing phase.} First, each uncoalesced vertex $v$ notifies every vertex $w$ in its (inclusive) $kT'$ hop neighberhood that $w$ will be active in the stage.  Each active vertex also learns its shortest path length from an uncoalesced vertex.  Next, each active vertex generates its portion of the random strings $r_{-T'},...,r_{-(T'/2)-1}$.  
These correspond to functions $f_{-(T'/2+1)}, f_{-(T'/2+2)}, \ldots, f_{-T'}$.
Note that the random strings $r_i$, $i > -(T'/2)-1$, are retained from the previous stage.

\noindent
\textbf{Main phase.} Every active vertex begins the phase by resetting its label set to $\{l\in L: b_v(l)>0\}$.  The main phase is composed of $T'$ steps $\{0,...,T'-1\}$.  

At step $i$, every active vertex $v$ with distance at most $k(T'-1-i)$ from an uncoalesced vertex collects the part of $r_{-T'+i}$ as well as the current label set from every vertex in its $k$ hop neighborhood.  Now $v$ has enough information to compute the label $f_{-T'+i}$ assigns to $v$ for every labeling of its $k$ hop neighborhood.  $v$ can now update its label set to be the union of the label it is assigned by $f_{-T'+i}$ for every \textit{possible} labeling of its $k$ hop neighborhood.  A labeling is only possible if every vertex is given a label from its label set.  

At the end of the phase, any vertices that have singleton label sets output the single label in their label set.

\subsection{Analysis}
At the end of the main phase, every vertex that was uncoalesced at the beginning of the stage has computed an upper bound on the set of labels that it could possibly be assigned by $F_{-T'}^0$.  Now consider a vertex $v$ that has a singleton label set at the end of the main phase.  This means that regardless of the input labeling, $F_{-T'}^0$ assigns $v$ a single label $l$.  Since functions are composed in a ``backwards'' order, we also know that $F_{-T''}^0$ assigns $v$ the single label $l$ for all $T''>T'$.  Once every vertex has output a label, we see that $F_{-T'}^0$ is singular.  Furthermore, the labeling output by the vertices is the unique labeling in the image of $F_{-T'}^0$.  Therefore the output labeling is exactly drawn from $\pi$, the desired sampling distribution, since $\pi$ is also the stationary distribution of the Markov chain by Theorem \ref{detailed-balance}.

The algorithm runs in $O(k \cdot T^*)$ rounds in the \local model, where $T^*$ is the smallest value such that $F_{-T^*}^0$ is singular.  When $k=1$ and the set of labels for the vertices has constant size, the algorithm runs in $O(T^*)$ rounds in the \congest model.  In theory, our algorithm requires exponential work per machine; however, for some practical examples such as the hardcore model and weighted dominating sets only polynomial work is required per machines.  This is because we can determine all of the possible labels of a vertex with a simple rule.

We now prove a theorem that shows that in some cases $T^*=O(\log n)$ with high probability.  This theorem statement is similar to Theorem 2 from \cite{BoundingChain1}; however, our method of proof is slightly different.  Following the lead from that paper (they credit some ideas to Murdoch and Green \cite{multigammacoupling}), we choose $\gamma$ to be a lower bound, over every vertex $v$, on the probability that $f_i$ assigns $v$ a single label $l\in L$ regardless of the input labeling, conditioned on $v$ being marked active.  We also choose $\beta$ to be an upper bound, over every vertex $v$, on the probability that for any two labelings $l_1,l_2$ with $l_1(v) = l_2(v)$ we have $(f_i(l_1))(v)\neq (f_i(l_2))(v)$, conditioned on $v$ being marked active.  Intuitively, $\gamma$ and $\beta$ describe the likelihood of a vertex moving towards or respectively away from a singleton label.
\begin{theorem}
\label{runtimetheorem}
The distributed bounding chain has $T^*=O(\frac{1}{p\gamma-\Delta_k p\beta}\log n)$  with high probability if $\gamma>\Delta_k \beta$, where $\Delta_i$ is the number of nodes in the largest (exclusive) $i$ hop neighborhood of the graph and $k$ is the maximum diameter of any constraint.
\end{theorem}
\begin{proof}
We will use the standard trick of considering the value of $T$ needed for $F_0^T$ to be singular, instead of directly considering $T^*$.  While 'Coupling to the Future' is not a valid sampling technique, the distributions of time needed for singularity are the same for $F_0^T$ and $F_{-T}^0$.  Let $Y^v_t$ be the indicator random variable for whether vertex $v$ is always given a singleton label set by $F_0^t$.  Let $Y_t$ be the sum of all the $Y^v_t$.  Assuming that the set of possible labels for each vertex contains at least two elements, we have $Y_0 = n$.  The only way for a vertex with a singleton label set to grow in size, is if it shares a constraint with a vertex with a non-singleton label set.  We can now see
$$E[Y_{t+1}] \leq Y_t+Y_t\Delta_k p\beta-Y_t p\gamma=Y_t(1+\Delta_k p\beta-p\gamma)=Y_t(1-(p\gamma-\Delta_k p\beta)).$$
We want to have $p\gamma-\Delta_k p\beta>0$, which is equivalent to the condition $\gamma>\Delta_k \beta$. Setting $\alpha = (p\gamma-\Delta_k p\beta)$, we now have, for $0<\alpha<1$,
$$E[Y_{t+1}] = \sum_{i=0}^n E[Y_{t+1} | Y_t = i] P(Y_t = i)\leq \sum_{i=0}^n (1-\alpha)i P(Y_t = i)=(1-\alpha)E[Y_t].$$
By induction, it follows that $E[Y_t]\leq (1-\alpha)^tn$.  By Markov's inequality, for $T=O(\frac{1}{\alpha}\log n)$, it follows that $P(Y_T\geq 1)\leq \frac{1}{n}$, which completes the proof.
\end{proof}

\begin{corollary}
When $\gamma-\Delta_k \beta > r$ for some constant $r>0$, the algorithm runs in $O(k\log n)$ rounds in the \local model.  Furthermore, if $k=1$ and $L$ has finite size, the algorithm runs in $O(\log n)$ rounds in the \congest model.
\end{corollary}



\section{Results for specific CSPs}
\begin{theorem}
An $O(\log n)$ round \congest algorithm exists for sampling from the hardcore model when the fugacity $\lambda\leq\frac{\alpha}{\Delta}$ for any constant $\alpha<1$.
\end{theorem}
\begin{proof}Note that a choice of $\gamma$ is $\frac{1}{1+\lambda}\big(1-\Delta p\frac{\lambda}{1+\lambda}\big)$, since a vertex will always accept a proposal to leave the independent set if no neighboring vertex is marked and proposing to join the independent set.  A choice of $\beta$ is $\frac{\lambda}{1+\lambda}$, since the label set of a vertex can only grow if the vertex is proposing to join the independent set.  Finally, $\Delta_k=\Delta$.  Thus we need $\gamma-\Delta \beta > r > 0$, for some constant $r$.

Note that $$\frac{1}{1+\lambda}\big(1-\Delta p\frac{\lambda}{1+\lambda}\big)-\Delta\frac{\lambda}{1+\lambda}\geq \frac{1}{1+\lambda}\big(1- p\frac{\alpha}{1+\lambda}\big)-\frac{\alpha}{1+\lambda}=$$
$$\frac{1-\alpha}{1+\lambda}-\frac{p\alpha}{(1+\lambda)^2}\geq \frac{1-\alpha-p\alpha}{1+\lambda}\geq \frac{1-\alpha(1+p)}{1+\alpha}$$
Choosing $p$ small enough we are done.
\end{proof}

\begin{theorem}
\label{WDS}
An $O(\log n)$ \congest algorithm exists for sampling weighted dominating sets when $\lambda \geq \alpha\Delta^2$ for any constant $\alpha>1$.
\end{theorem}
\begin{proof}
While there are now constraints that are not unary or binary, we are still able to run the algorithm in the \congest model.  This is because necessary information can be aggregated at the vertex at the center of each constraint and then distributed to the other vertices of the constraint.

Similar to the hardcore model we can choose $\gamma = \frac{\lambda}{1+\lambda}\big(1-\Delta_k p\frac{1}{1+\lambda}\big)$ and $\beta=\frac{1}{1+\lambda}$.  Note that here $\Delta_k = \Delta+\Delta(\Delta-1)=\Delta^2$.  Again it is sufficient for $\gamma-\Delta_k \beta > r > 0$, for some constant $r$.

$$\frac{\lambda}{1+\lambda}\big(1-\Delta_k p\frac{1}{1+\lambda}\big)-\Delta_k \frac{1}{1+\lambda}\geq \frac{\alpha\Delta_k}{1+\alpha\Delta_k}\big(1-\frac{\Delta_k p}{1+\alpha\Delta_k}\big)-\frac{\Delta_k}{1+\alpha\Delta_k}=$$
$$\frac{\Delta_k}{1+\alpha\Delta_k}\bigg(\alpha\big(\frac{1+\Delta_k(\alpha-p)}{1+\alpha\Delta_k}\big)-1\bigg)\geq \frac{1}{1+\alpha}\bigg(\alpha\big(\frac{\alpha-p}{\alpha}\big)-1\bigg)=\frac{a-p-1}{1+a}$$
We are done as long as $p$ is sufficiently small.
\end{proof}

\begin{remark}The simplified Ising model with parameter $\beta>1$ remains monotone in the distributed setting; see \cite{CFTP} for a more detailed explanation in the sequential setting. Our algorithm can be modified to give samples from the Ising model.  Instead of keeping track of potential labels for each vertex, each vertex keeps track of its label when the input of $F_{-T}$ is $\top$ and $\bot$.  When they are equal the vertex outputs this label.  We don't have a bound on the runtime of this approach; however, combining the analysis of monotone CFTP \cite{CFTP} and the remark of \cite{FischerGhaffariColoring}\footnote{They mention that the Dobrushin condition \cite{dob} is enough for fast mixing of their chain, however we cannot guarantee that we are using the generalization they mention and also they don't provide a proof of their claim.} may imply an efficient, $O(\log^2 n)$ round, \congest algorithm for some range of $\beta$.
\end{remark}

\section{Conclusion}
We have shown that for certain weighted local CSPs, CFTP combined with the bounding chain technique allows efficient exact sampling in the \local and sometimes even the \congest model. 
Previous work \cite{DistributedJVVPaper,LLLPaper} that achieved results for exact distributed sampling used very different techniques, so a conceptual contribution of our paper is showing that CFTP and bounding chains should also be added to our toolkit for exact distributed sampling.  We wish to highlight two open questions suggested by this work.  

Does there exist a logarithmic-round \congest or \local algorithm that exactly samples uniform colorings from a palette of size $O(\Delta)$?  
In the sequential setting, one of the original papers describing bounding chains showed that there is a polynomial-time exact coloring algorithm when the number of colors is $\Theta(\Delta^2)$ \cite{BoundingChain2}. Very recent work on bounding chains in the sequential setting has shown that uniform exact sampling of colorings is possible in polynomial time using only $\Theta(\Delta)$ colors \cite{BhandariChakrabortyColorBoundingChain,NewestColorBoundingChain}.

Compared to the \congest model, many problems can be solved much faster in ``all-to-all'' communication models such as the \congestedclique model \cite{lotker2005mstJournal,ghaffari16_mst_log_star_round_conges_clique,HegemanPPSSPODC15,jurdzinski18_mst_o_round_conges_clique}. A question suggested by recent work on distributed sampling is whether much faster distributed sampling algorithms -- exact or approximate -- can be designed for the \congestedclique model.
%
%
%
%
\bibliographystyle{splncs04}
\bibliography{sample-base}
\end{document}